\documentclass[sigplan,screen,authorversion]{acmart}
\usepackage{subcaption}
\usepackage{todonotes}
\usepackage{siunitx}
\usepackage[nolist]{acronym}
\AtBeginDocument{%
  \providecommand\BibTeX{{%
    \normalfont B\kern-0.5em{\scshape i\kern-0.25em b}\kern-0.8em\TeX}}}

\copyrightyear{2023}
\acmYear{2023}
\setcopyright{none}
\acmBooktitle{The 31st International Conference on Real-Time Networks and Systems (RTNS 2023), June 7--8, 2023, Dortmund, Germany}
\acmDOI{10.1145/3575757.3593644}
\acmISBN{978-1-4503-9983-8/23/06}

\newtheorem{theorem}{Theorem}[section]
\begin{document}

\title{On the Validity of Credit-Based Shaper Delay Guarantees in Decentralized Reservation Protocols}


\author{Lisa Maile}
\affiliation{%
  \institution{Computer Science 7\\Friedrich-Alexander-Universität}
  \city{Erlangen-Nürnberg}
  \country{Germany}}
\email{lisa.maile@fau.de}

\author{Dominik Voitlein}
\affiliation{%
  \institution{Computer Science 7\\Friedrich-Alexander-Universität}
  \city{Erlangen-Nürnberg}
  \country{Germany}}

\author{Alexej Grigorjew}
\affiliation{%
  \institution{Chair of Communication Networks\\Julius-Maximilians-Universität}
  \city{Würzburg}
  \country{Germany}}

\author{Kai-Steffen J. Hielscher}
\affiliation{%
  \institution{Computer Science 7\\Friedrich-Alexander-Universität}
  \city{Erlangen-Nürnberg}
  \country{Germany}}

\author{Reinhard German}
\affiliation{%
  \institution{Computer Science 7\\Friedrich-Alexander-Universität}
  \city{Erlangen-Nürnberg}
  \country{Germany}}

\begin{abstract}
Resource reservation is a fundamental mechanism for ensuring quality of service in time-sensitive networks, which can be decentralized by using reservation protocols. 
In the Ethernet technology Time-Sensitive Networking, this has been proposed in conjunction with the Credit-Based Shaper. For the reservation, the standards assume a maximum worst-case latency bound at each hop.
However, we will show through formal analysis and simulation that these worst-case latency bounds are not safe.
To face this, we propose an extension to the current standards to allow the reservation of time-sensitive traffic with reliable latency guarantees. The effectiveness of our approach is demonstrated through simulations of both synthetic and industrial networks.
Finally, by providing additional information about neighboring devices, we could further increase the maximum reservable traffic by up to 20\% in our test cases.
\end{abstract}



\maketitle 
\vspace{-4mm}
\section{Introduction}
With the increasing popularity of digital \ac{AV} applications, there has been a growing demand for more advanced transmission mediums to support them. Applications that combine \ac{AV} streams - such as lip synchronization, broadcasting, gaming, and virtual reality - require precise synchronization between the streams. Traditional networks, like standard Ethernet, cannot meet these strict \ac{QoS} requirements due to potentially high delays and jitters. As a response, the IEEE \ac{AVB} task group developed new standards for Ethernet. These standards have introduced, i.a., the \ac{CBS} - a new forwarding algorithm~\cite{Qav}, and provided decentralized reservation protocols, such as the \ac{SRP}~\cite{Qat}. In 2012, \ac{AVB} was renamed to \ac{TSN} to reflect that low delay and jitter are necessary for a wide range of applications.

Decentralized resource reservation in \ac{TSN} allows systems to self-organize while keeping strict \ac{QoS} requirements, such as maximum end-to-end delays and no packet loss. A reservation protocol calculates the maximum per-hop latencies on a flow's path and reserves resources for the flow if its end-to-end delay requirement is met. Although originally designed for \ac{CBS} networks, new \ac{TSN} draft standards also promise the integration of other network schedulers~\cite{Qdd}. As a result of the ongoing standardization process, decentralized admission control is a highly active area of research.

\ac{TSN} standards propose three different methods for computing the maximum local latency at each bridge for decentralized admission control in \ac{CBS} networks in~\cite[Section~6.6]{BA},~\cite[Section~L.3.1.1]{802Q}, and~\cite{plenary}, where~\cite{plenary} is referenced in the 802.1Q standard, see~\cite[p.~1569]{802Q}.
In this paper, we will demonstrate that all the equations provided by the standards are unsafe, meaning that they do not provide a valid upper bound on the per-hop latency. We will further prove that the current decentralized admission process cannot derive safe upper bounds. To face this, we propose \ac{SSRP} as an extension to the existing standards. \ac{SSRP} allows for reliable worst-case delays using the mathematical framework Network Calculus. We prove our results through formal analysis and use simulation to provide counterexamples to the standards' equations and to evaluate our approach. To the best of our knowledge, we are the first to show that the latency calculations in the standards are unsafe. To sum up, our main contributions are:
\begin{itemize}
    \item Simulation of the delay bounds from the standards to show that they do not cover the worst case
    \item Analytical proof that the queuing delay in \ac{CBS} networks with the given protocol cannot be upper bounded
    \item Introduction of \ac{SSRP} to finally allow for decentralized delay guaranteeing \ac{CBS} networks
\end{itemize}

The remainder of this paper is organized as follows. Section~\ref{sec:related} presents related work. We provide an overview of decentralized admission control in \ac{TSN}, including the existing reservation protocols and the \ac{CBS} forwarding mechanism, in Section~\ref{sec:fundamentals}. In Section~\ref{sec:state-of-the-art}, we introduce the three latency calculations proposed by the \ac{TSN} standards. Section~\ref{sec:infinite} will prove that the queuing delay in \ac{CBS} networks can become unbounded. In Section~\ref{sec:ourapproach}, we will propose and prove \ac{SSRP} which provides safe guarantees. Finally, Section~\ref{sec:evaluation} presents counterexamples for the standards' equations and evaluates our new solution using simulation, before Section~\ref{sec:conclusion} concludes the paper.

\section{Related Work}
\label{sec:related}
Grigorjew et al.~\cite{grigorjew_decentralSP_2020} proposed a decentralized admission control scheme, which is the first to allow decentralized reservation for Strict Priority networks. Our work extends the approach of~\cite{grigorjew_decentralSP_2020} to \ac{CBS} networks, which can benefit from the shaping effect of \ac{CBS} queues. In~\cite{grigorjew_decentralSP_2020}, the authors determine the maximum number of packets that can arrive at each queue and utilize these packets to account for the per-hop delay. In addition to this, we also consider the impact of both link shaping and \ac{CBS} shaping on packet arrival, allowing us to model that packets must arrive in a serialized manner rather than all at once.

Boiger~\cite{boiger} presented a counterexample to a statement on the maximum delay in \ac{CBS} networks at the IEEE 802 Plenary Meeting. The standard states that "\SI{2}{\milli\second} [...] for SR Class A can be met for 7 hops of \SI[per-mode=symbol]{100}{\mega\bit\per\second} Ethernet if the maximum frame size on the LAN is 1522 octets"~\cite{BA}. Boiger showed an example where this \SI{2}{\milli\second} end-to-end delay is violated. However, the standard delay formulas have not been discussed. We will compare these formulas and identify the root cause of the decentralized delay violations. To sum up, no valid delay formula for decentralized \ac{CBS} network exists until now.

\begin{table}[t]
\centering
\caption{Notation}\vspace{-2mm}
\resizebox{\linewidth}{!}{%
\begin{tabular}{|l|l|}
\hline
Variable & Definition \\ \hline\hline
\textbf{Topology / Routes} & \\	
		$C$ & link capacity\\
		$Q$ & number of CBS priorities\\
  $\mathcal{L}^-$ & set of input links \\
		$Q^-_l$ & set of \ac{CBS} queues at input link $l$\\
  $\Phi_f^q$ & queues on path of flow $f$ before queue $q$\\
  \textbf{Flows} & \\	
		$F_l$, $F_{l,q}$ & set of flows arriving from link $l$ (and queue $q$) \\
  $\mathit{CMI}$ & sending interval of a flow \\
		$\mathit{MFS}$ & max. packet size of a flow \\
		$\mathit{MIF}$ & max. packets per interval of a flow \\	
		
		\textbf{Frame Sizes} & \\	
        $L_{f}, l_{f}$ & max./min. packet size of one flow \\
		$L$, $l$ & max./min. packet size at a queue\\
		$L_{max}$, $L_{min}$ & max./min. packet size in the network\\
  
  \textbf{Delays} & \\	
  		$\overline{D}$, $\underline{D}$ & max./min. hop delay (pre-configured) \\
          	$D$ & worst-case hop delay (current state) \\
            $\overline{\mathcal{D}}_f$, $\underline{\mathcal{D}}_f$ & accumulated max./min. latency of flow $f$\\
\textbf{CBS} & \\ 	
		$idSl$, $sdSl$  & idleSlope / SendSlope\\
  		$c$ & CBS credit  \\
		$c_{max}$, $c_{min}$ & max./min. CBS credit  \\
		
		\textbf{Network Calculus} & \\	
        $R(t)$, $R^*(t)$ &  cumulative incoming/outgoint traffic\\
		$\alpha$, $\alpha_f$ & queue arrival / flow arrival curve \\
		$b^f$, $\hat{b}$ & packets in a burst from flow $f$ / cross-traffic\\
		$\beta_{R,T}$ & service curve with latency $T$ and rate $R$  \\
		$\sigma_{l}$, $\sigma_{CBS}$ & shaping curve of link / CBS queue \\

        \textbf{Standards} & \\	
        $t_{L} = \frac{L}{C}$ & Transmission time of packet with size $L$ \\
        $t_{proc}$ & Processing delay \\
        $t_{prop}$ & Propagation delay \\
        $t_{sf}$ & Store-and-forward delay \\
        $t_{inQueue}$ & Input queuing delay  (typically not present) \\
        $t_{oct} = 8\mathrm{bit} / C$ & Octet transmit time \\
  \hline
\end{tabular}%
}
\label{tab:notation}
\end{table}

We use \ac{NC} - as it is a well-established delay analysis framework - to offer local guarantees which remain valid even when new flows are added to the network. Details of the delay analysis of \ac{CBS} using \ac{NC} can be found in~\cite{queck, framework, zhao_improving_avb_2018, zhao_latency_2020}. A complete overview of \ac{NC} results for \ac{TSN} has been presented in~\cite{maile_network_2020}. All of these works analyze the flow delays for static flow reservations, while we derive local upper bounds which remain valid even after new flow registrations.  In~\cite{maile_journal_2022}, Maile et al. use \ac{NC} for \ac{CBS} networks for the configuration of delay guarantees with a central network controller while the network is running. Contrary, our paper presents a solution for decentralized network setups.

\section{Fundamentals}
\label{sec:fundamentals}
\ac{TSN} has been extended and enhanced significantly over the last years.
We will explain the current state-of-the-art concepts in the following section.

\subsection{Notation}
We model a network as a directed multigraph, $G = (V,E)$, where each vertex $V$ (hereafter called \textit{node}) represents an \ac{ES} or a bridge. Each edge, $E$, represents a single \ac{CBS} output queue. We identify each output queue with priority $p$ uniquely with the tuple $(u,v,p)\in E$ with $0\le p < Q$. This means that there are $Q$ edges between nodes $u$ and $v$. Physical links are denoted by $(u,v)$. Table~\ref{tab:notation} provides a list of all variables for our models. Each variable in Table~\ref{tab:notation} can have a superscript, to identify the corresponding priority $(p)$, link $(u,v)$, or queue $(u,v,p)$, e.g., $\bullet^{(7)}$ for priority 7. We use the term stream and flow interchangeably. 
We use $D$ for the worst-case per-hop delay and $\mathcal{D}_f$ for the cumulative path delay of a flow until a specific hop (sum of all previous delays); 
with over- and underlines for upper and lower bounds.
To reduce the notation complexity, we assume the same link capacity $C$ at each output port. 

\subsection{Decentralized Admission Control}

Decentralized admission control does not rely on
a central configuration unit during operation, but on autonomous decisions
of every involved networking device. As a result, they require concise communication to provide deterministic guarantees.

\textbf{Configuration:}
Distributed admission control relies on budgets for available resources on each node, i.e., the available bandwidth per traffic class or latency budgets.
Admission control ensures that a given threshold is not exceeded
by new reservations for each such resource type.
Typically, these thresholds can either be pre-configured as default values,
configured manually by a network operator, or they can be defined by a
\ac{NMS}.
Note that, unlike a central control unit, the \ac{NMS} is not involved in the reservation process. After the initial configuration, the networking units are operating autonomously.

\textbf{Communication:}
Generally, subscription-based information exchange protocols work in four steps.
(1) The source of the information (talker) advertises the availability of data,
which can be identified by tags or by individual IDs.
(2) The network elements distribute these advertisements throughout the network.
(3) End devices (listeners) receive the advertisements. If they are interested in the particular data stream,
they send a subscription to the talker.
(4) The talker starts to transmit the data, and the network elements forward it to all interested listeners.

For admission control, these advertisements must contain information that allows to calculate the required
resources for each flow.
Typically, this includes some kind of traffic specification, e.g., bandwidth requirements.
In a decentralized protocol, it is important to understand that every node has its own local view
of the subscriptions in the network.
Ideally, each node only needs to know about the subscriptions whose path
includes them. Otherwise, all subscriptions must be broadcasted through the entire network, which creates
an overhead that scales poorly with network size.

For real-time networking, worst-case latency is a type of resource under admission control.
However, new subscriptions do not only influence their own path,
but as a result of interference, they cause additional latency in other parts
of the network. Nevertheless, existing reservations should remain valid after accepting
new reservations.
This can be addressed by using shapers that prevent latency cascading,
e.g., Asynchronous Traffic Shaping~\cite{specht_urgency-based_2016,Qcr}, or by using pre-configured latency thresholds
as upper bounds for possible interference, which is independent of the current reservations~\cite{grigorjew_decentralSP_2020}.
The latter is proposed in this work, which requires that the accumulated latency thresholds are communicated along with the advertisement, and that the individual nodes ensure that these thresholds remain valid.

\subsection{Reservation Protocols}

Reservation protocols for specific \ac{QoS} requirements exist for more than 24 years.
In 1997, RFC 2205~\cite{rfc2205} introduced the Resource Reservation Protocol (RSVP) for layer~3 signaling of \ac{QoS} requirements.
RSVP allows several types of \ac{QoS} -- for deterministic scenarios, it mostly relied on IntServ~\cite{rfc2210} for reshaping individual flows.
The talkers use a token bucket \ac{TSpec} which includes bucket size, bucket rate, and peak data rate (e.g., link speed).
In addition, the reservation signaling includes accumulated intermediate computation results that shall enable decentralized latency estimation.
While technically appealing, RSVP/IntServ was only used reluctantly due to the expensive per-stream shaping operation.
When RSVP/IntServ is used without support by the intermediate layer~2 devices, their configuration becomes difficult.

Later in 2006, the IEEE 802.1 task group started their work on a layer~2 signaling protocol for deterministic \ac{QoS} requirements.
In 2010, \ac{SRP} was standardized, which mainly relies on the
Multiple Stream Registration Protocol (MSRP) for signaling.
The SRP talker TSpec consists of a maximum frame size and a maximum number of frames during a pre-defined
\ac{CMI} \cite[Sec.\,35]{802Q}. The \ac{CMI} was typically fixed at \SI{125}{\micro s}
and \SI{250}{\micro s} for the highest traffic classes.
In addition, \ac{SRP} included an accumulated latency field in its signaling. However, this field was
only intended to be used for end-to-end latency estimation by the listener.
SRP has not specified any logic that prevents bridges from exceeding their pre-configured
latency during operation.
Most notably, SRP itself has not specified a latency model but left latency computation open
for other standards, such as the audio/video profile 802.1BA~\cite{BA}.
These latency models were inaccurate in general operation, as demonstrated in Sec.\,\ref{sec:state-of-the-art}.

Finally, in 2018, the IEEE TSN group started working on the successor of SRP, the \ac{RAP}~\cite{Qdd}. It is still under active development
and currently available in draft~0.6. In this context, notable improvements include:
(1) using a more flexible TSpec, e.g., token bucket or variable CMIs;
(2) including more intermediate fields in signaling, such as accumulated minimum latency,
for improved latency model accuracy;
(3) inclusion of mandatory resource bound checks inside the signaling procedure,
e.g., to ensure that pre-configured per-hop latency thresholds remain valid.
Currently, the group is working on concrete models for latency computation
to provide a safe, self-contained QoS signaling protocol.
In addition, heterogeneous deployments are an important aspect of RAP,
which necessitates more communication between neighboring devices,
e.g., sharing shaper parameters.
Latency models, such as this work, could provide valuable input
for this standardization process.

\subsection{Credit-Based Shaper}
\ac{CBS} was originally proposed in the IEEE 802.1Qav standard~\cite{Qav} for audio and video communications.
Fig.~\ref{fig:port} shows the structure of an output port implementing \ac{CBS}. Each output port consists of up to eight queues, with a priority level ranging from 0 to 7, with 7 being the most important priority. Packets are placed in the queue corresponding to their priority, so multiple flows share the same queue. The mapping of priorities to queues depends on the available number of queues in each node. Packets in the same queue are served in a First-In-First-Out manner. Each \ac{CBS} queue (denoted as $(u,v,p)$) has a credit value, $c^{u,v,p}(t)$, and packets are only eligible for transmission if the credit value is greater than or equal to zero. If multiple packets are eligible at the same time, they are transmitted according to their priority. The standard sometimes refers to classes instead of priorities, e.g., class~A instead of priority~7, but we will use the term priority in the following. We explain the functionality of \ac{CBS} with the example shown in Fig.~\ref{fig:credit}.

\textbf{Phase (a)}: At time $t=0$, there are three packets with priority 7 and one packet with priority 6 waiting to be transmitted. For illustration, we assume that a \ac{BE} packet is under transmission when the four \ac{CBS} packets arrive.	
The credit $c^{u,v,p}(t)$ for each \ac{CBS} queue increases with the rate called idleSlope $idSl^{u,v,p}$. The credit increases while scheduled traffic cannot be sent due to an ongoing transmission of other queues, or while the credit is negative. E.g., in Fig.~\ref{fig:credit} both queues need to wait for the transmission of the \ac{BE} packet at $t=0$ and, thus, their credit increases. 
	
\textbf{Phase (b)}: After the \ac{BE} packet, both CBS queues have a positive credit and eligible packets, so the priority is used to decide on transmission. 
During the transmission of packets, the credit of the sending queue decreases with the rate called SendSlope with $sdSl^{u,v,p} = idSl^{u,v,p} - C$
	
\textbf{Phase (c)}: In a worst-case scenario, priority 7 begins with the transmission of a packet with maximum size at the moment its credit reaches 0. This delays the transmission of priority 6 the most.
	
\textbf{Phase (d)}: After transmitting three packets of priority 7, the credit of this queue is negative, and priority 6 is allowed to send. New packets of priority 7 would be delayed until the credit $c^{u,v,7}(t)$ reaches at least 0 again.
	
\textbf{Phase (e)}: If a queue is empty but has a positive credit, the credit value is reset to zero.

The idleSlope of a \ac{CBS} queue determines both the maximum and minimum guaranteed bandwidth for that queue. Due to the non-negative credit constraint, \ac{CBS} queues introduce a per-queue shaping behavior up to a limited maximum bandwidth. As a result, average delays for lower priority and \ac{BE} traffic~\cite{flowshaping, regev_is_2016} are improved. In contrast, traditional schedulers, such as Strict Priority or Deficit Round Robin, do not implement a maximum bandwidth. Since \ac{CBS} is a queue-aggregated shaping mechanism, the shaping is much cheaper than the per-stream shaping of IntServ. With the limited maximum bandwidth, all queues also have a minimum guaranteed bandwidth. 

\begin{figure}[t]
	\centering
	\includegraphics[width=0.85\linewidth]{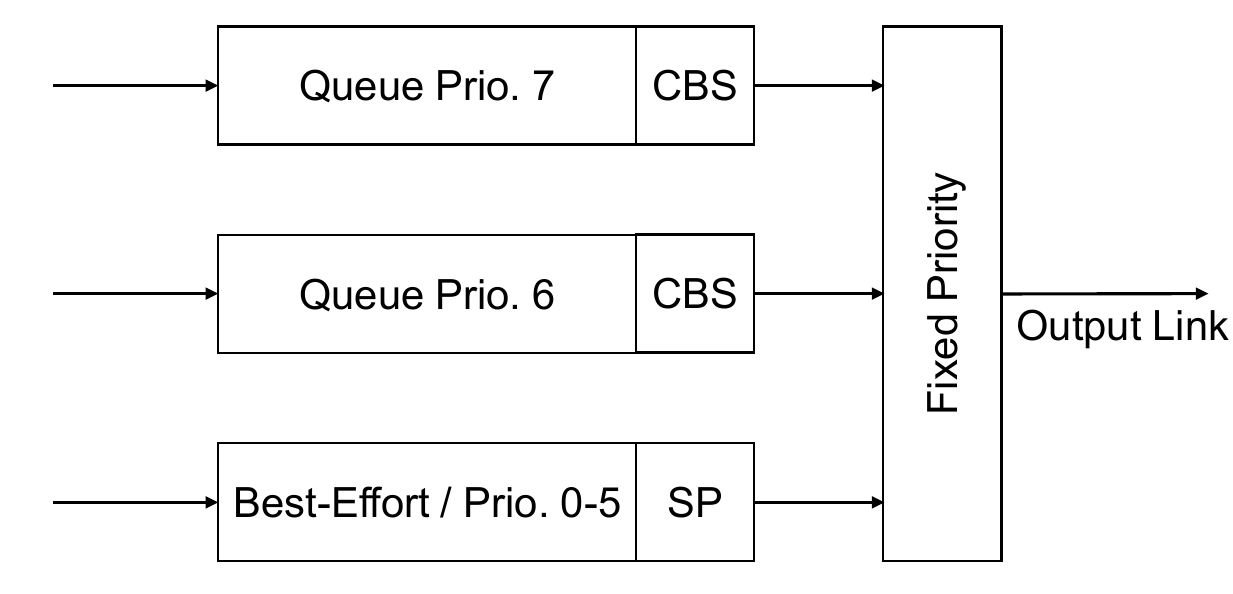}
	\caption{Example output port with CBS queues.}
	\label{fig:port}
\end{figure}
\begin{figure}[t]
	\centering
	\includegraphics[width=\linewidth]{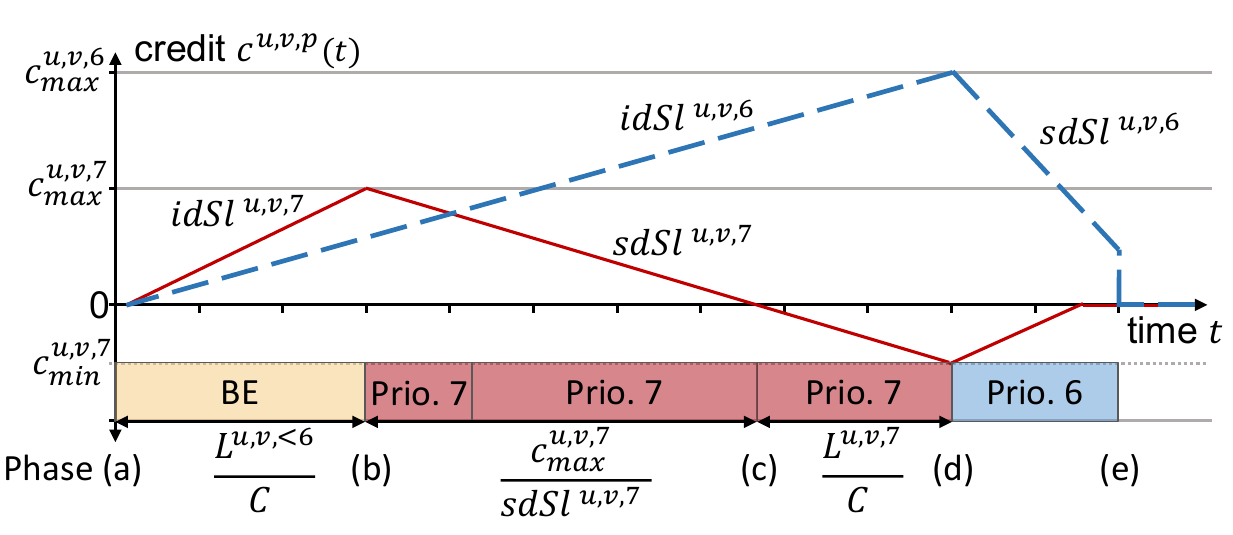}
	\caption{Credit evolution for two CBS queues with packet transmissions illustrated at the bottom.}
	\label{fig:credit}
\end{figure}

\section{State-of-the-Art Latency Calculation}
\label{sec:state-of-the-art}
\ac{TSN} standards have proposed three different methods for computing the maximum local latency at each bridge for decentralized admission control in \ac{CBS} networks. The packet sizes include preamble, \ac{SFD}, and \ac{IPG}. 

    

\subsection{IEEE 802.1BA}

The 802.1BA standard assumes a scenario where a maximum length packet and all other packets of one \ac{CMI} with the same priority delay the packet of interest. The calculation is only given for the priority~7 \cite[Section~6.6]{BA}:
\begin{align}
    \overline{D} &= t_{proc} + \underbrace{t_{L_{max}} \vphantom{\left(\frac{X}{X}\right)}}_{\text{other pr.}} + \underbrace{\left(\frac{idSl}{C} \cdot \mathit{CMI} - t_{L_{{FoI}}}\right) \cdot \frac{C}{idSl}}_{\text{same priority}} + t_{L_{FoI}-IPG}
\end{align}
This assumes that, in each \ac{CMI}, traffic with the same priority is received with the link speed $C$ for a fraction of $(idSl/C)$ and is forwarded with the rate $idSl$.

\subsection{IEEE 802.1Q}
802.1Q-2018 Annex L.3 is the only standard that provides delays for more than one priority~\cite{802Q}. It considers input queuing, interference, transmission, propagation, and store-and-forward delays:
\begin{equation}
    \overline{D} = t_{inQueue} + t_{int} + t_{L} + t_{prop} + t_{sf}
\end{equation}
The input queuing and store-and-forward delays are (part of) the processing delay of the hardware.
The interference delay $t_{int}$ is the sum of queuing, fan-in, and permanent buffer delays ($t_{queue}$, $t_{fan-in}$, and $t_{perm}$ respectively). The queuing delay is the time it takes to transmit one packet with maximum size and all higher-priority packets~\cite{802Q}:
\begin{equation}
    t_{queue} = \begin{cases}
    L_{max} / C & \text{for prio. 7} \\
    (L_{max} + L^{(7)}) / (C - idSl^{(7)}) & \text{for prio. 6} \\
\end{cases}
\end{equation}
802.1Q-2018 defines the fan-in delay $t_{fan-in}$ as "delay caused by other frames in the same class [...] that arrive at more-or-less the same time from different input ports"~\cite[p.~1951]{802Q}. 
The procedure is explained in~\cite[Section~L.3.1.2]{802Q}. Since the packets of a fan-in burst reside in the buffers for some time, they cause further delays until they leave the system, which is reflected by the permanent buffer delay $t_{perm}$. Because the buffered packets are the result of a fan-in, there is only either fan-in or permanent buffer delay. 




\subsection{IEEE 802.1Q - Plenary Reference}
The following equations are the result of a plenary discussion, which is referred to in the 802.1Q standard. All calculations are for priority~7 and \SI[per-mode=symbol]{100}{\mega\bit\per\second} links\footnote{Equations for 
\SI[per-mode=symbol]{1}{\giga\bit\per\second} links have been added afterwards by John Fuller, but have not been validated by the committee.}. 

The calculation uses a scenario where a maximum-length packet and frames from all other input ports are placed in the queue at the same time. Then another set of priority~7 packets arrives together with the packet of interest, which is the last frame in the queue.
Unlike the previous two approaches, this calculation uses the number of octets and not the time of each contributing factor~\cite{plenary}:
\begin{equation}
    \overline{D} = \Biggl(\underbrace{L_{max}\vphantom{\left(\frac{X}{X}\right)}}_{\text{other pr.}} + \underbrace{2 \cdot \left(R_{max} - L_{FoI}\right) - \left\lceil\frac{R_{max} - L_{FoI}}{N}\right\rceil}_{\text{same priority}} + L_{FoI} \Biggr) \cdot t_{oct}
\end{equation}
With maximum reserved octets:
\begin{equation}
    R_{max} = \left\lfloor \frac{\mathit{CMI}}{t_{oct}} \cdot \frac{idSl}{C}  \right\rfloor,~N = \min\left(|\mathcal{L}^-|~, \ \left\lfloor \frac{R_{max} - L_{FoI}}{L_{min}} \right\rfloor\right)
\end{equation}









\section{Infinite Latency Bound}
\label{sec:infinite}
In this section, we provide a formal proof that the local bridge latency cannot be upper bounded with the current standard procedure. A simulation which provides counterexamples to the existing standards' equations is given in Section~\ref{sec:evaluation}.

\begin{figure*}[t]
	\centering
	\includegraphics[width=0.95\linewidth]{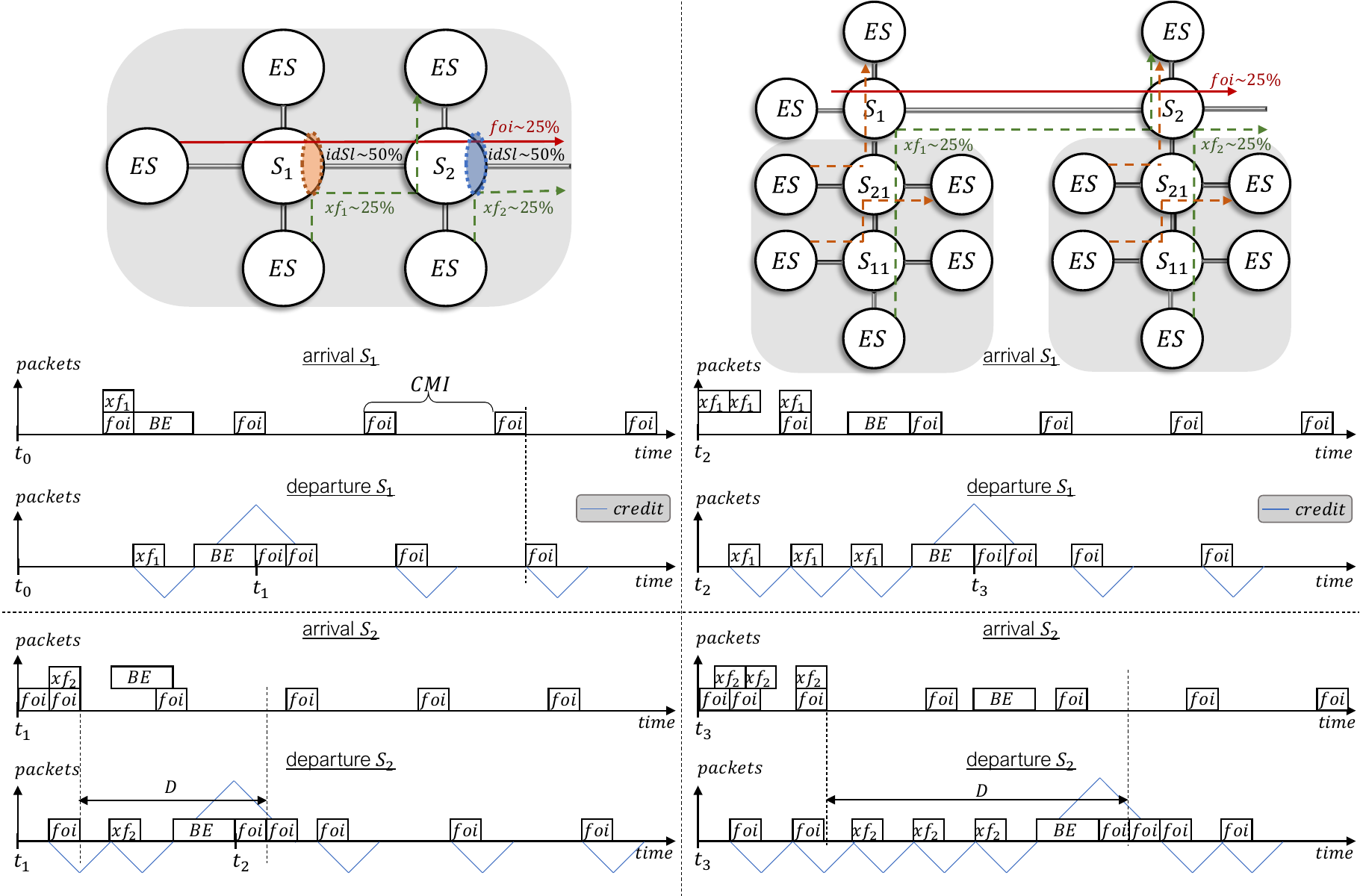}
	\caption{Transmission behavior with burstiness increase, assuming interfering traffic and store-and-forward. \\ \textbf{Top}: Topology with priority~7 flows. \textbf{Center}: Arrival and departure at port of $S_1$. \textbf{Bottom}: Arrival and departure at port of $S_2$.\\ \textbf{Left}: Initial network (recursion level 0). \textbf{Right}: Recursive topology increases the burst of cross-flows (recursion level 1).}
	\label{fig:approach-illustration}
\end{figure*}

The standards assume that the bandwidth and idleSlope of all connected links are given at each bridge. In addition, the protocols distribute each flow’s \ac{TSpec}, which includes the stream ID, priority, sending interval, and data. It does not assume further information, e.g., about a flow's interferences on its path or the overall network topology. 

To meet the end-to-end delay requirements, each bridge adds a maximum local latency to the accumulated latency field of the reservation protocol. Thus, the local latency guarantee may not change when new flows register or deregister afterwards. 
The standards assume a maximum local latency can be guaranteed, solely by checking whether the idleSlope would be surpassed with a new flow’s reservation. This section will prove that further information is required. 

We start by creating a subnetwork (level 0 of recursion) to demonstrate that a \ac{FoI} can accumulate bursts of packets as it travels along its path. Then, we append this subnetwork, in a way that cross-flows accumulate this burst for interference with the \ac{FoI} (level~1 of recursion, compare topologies in Fig.~\ref{fig:approach-illustration}). This increases the \ac{FoI}'s burstiness again. We recursively append the new subnetwork, to model more bursty interfering flows (recursion level $n, n\to\infty$), resulting in an unbounded delay for the \ac{FoI}.
\begin{theorem}[Unbounded Local Delay]
\label{theo:unbounded_local_delay}
With the given information for the decentralized \ac{TSN} reservation protocols (idleSlope and \ac{TSpec}) the maximum local latencies in \ac{CBS} are unbounded.
\end{theorem}
\begin{proof}
The idleSlope denotes the maximum guaranteed fraction of the bandwidth averaged over a time period~\cite{Qav}, thus, \ac{CBS} packets are forwarded with the respective idleSlope. It is assumed that the reserved flow rates are independent of the configured idleSlopes~\cite{queck}, but do not surpass the idleSlope. For simplicity, we assume processing and propagation delay to be zero. We provide the proof for the highest priority only and omit the indices. Instead, we denote each variable with a recursion level $n$ and the hop ID $i$ as $\bullet^{i,n}$. The proof for the lower priorities is done analogously.

Let us denote the term \textit{burstiness} of a flow as the sum of data in packets that are transmitted consecutively with a rate $\ge idSl$. 
Assume the network provided on the left of Fig.~\ref{fig:approach-illustration}, with idleSlopes configured at 50\% of link speed and three flows transmitted at their source with 25\% of the link speed, as illustrated. We define the left topology as recursion level 0. The formulas below model the traffic behavior in this network. Let $D^{i,n}$ be an estimation of the local worst-case delay at hop $S_i$. By definition, we can define the accumulated maximum and minimum latency on the \ac{FoI}'s path before hop $S_i$ in store-and-forward networks as:
\begin{equation}
    \overline{\mathcal{D}}^{i,n} = \sum_{j=1}^{i-1}D^{j,n},\quad \text{and}\quad\underline{\mathcal{D}}^{i,n} = \sum_{j=1}^{i-1}\frac{l}{C}.
\end{equation}
The burstiness of a flow can increase until hop $S_i$ by
\begin{equation}\label{equ:proof_b_f}
    b_{FoI}^{i,n} = \Bigg\lceil \underbrace{\frac{\overline{\mathcal{D}}^{i,n}-\underline{\mathcal{D}}^{i,n}}{\mathit{CMI}}}_{(a)}+\frac{\Delta t}{\mathit{CMI}} \Bigg\rceil~,
\end{equation}
where $\Delta t$ is the time to transmit the burst which accumulated in (a). If we denote the burstiness of cross flows as $\hat{b}^{i,n}$, we can derive the worst-case per-hop delay:
\begin{equation}\label{equ:proof_D_i}
    D^{i,n} = \frac{b_{FoI}^{i,n}+\hat{b}^{i,n}}{idSl} + \frac{L_{max}+L}{C}
\end{equation}
In the left topology of Fig.~\ref{fig:approach-illustration}, $\hat{b}^{i,n}_{FoI}$ is upper bounded at each hop $S_i$ by the interference of one \ac{BE} packet with maximum size and one packet of the interfering cross flows, resulting in $D^{1,0}=L/idSl+(L_{max}+L)/C$ as illustrated at \textit{departure $S_1$} in Fig.~\ref{fig:approach-illustration}.
After crossing $S_1$, the \ac{FoI} exhibits a burstiness of two packets. In $S_2$, the interference is the same, but the burstiness of the \ac{FoI} increases further. The result is shown at the left-bottom of Fig.~\ref{fig:approach-illustration}. $D^{2,0}$ is limited by the path length of the \ac{FoI}. We will now show that even with a fixed path length, the delay can be unlimited. 

As we are not aware of the topology, we are also not aware of the path of interfering flows. As a result, we may assume a recursive topology as shown on the right of Fig.~\ref{fig:approach-illustration}. With the recursiveness, the cross flows themselves can have the burstiness which the \ac{FoI} experiences after traversing $S_2$ in the first network, meaning $\forall i, \hat{b}^{i,1} = b_{FoI}^{2,0}$. According to Eq.~\eqref{equ:proof_D_i}, the \ac{FoI}'s delay increases when compared to the first topology, which again results in an increased burstiness, according to Eq.~\eqref{equ:proof_b_f}. We consider only bursts from cross flows. Non-bursty packets are either lost or not sent. This is shown at the right of Fig.~\ref{fig:approach-illustration}. We can easily continue this recursive behavior, using the output after $S_2$ as increased cross-traffic interference in each recursion level:
\begin{equation}
    \hat{b}^{i,n} = b_{FoI}^{2,n-1}, \forall i, n=1,2,...
\end{equation}
As a result, the cross-flows can add up to infinite burstiness, surpassing the idleSlope for an undefined time and infinitely increasing the delay $D^{2,n}$, $n\to\infty$, of our \ac{FoI}. 
\end{proof}

Our simulations in Section~\ref{sec:eval_comparison} show that this increase in burstiness results in significant delays, even in small networks. 
The infinite burstiness can occur at any network where the sum of idleSlopes from the input links can surpass the idleSlope of the output link, which typically is the case for most networks.


\section{Reservation with Bounded Latency}
\label{sec:ourapproach}

\subsection{General Approach}
To ensure safe configurations, we introduce our approach, \ac{SSRP}, which implements the following new concepts to the existing reservation process: (1) We define pre-configured delay budgets for each hop and check, that the actual worst-case delay does not surpass this budget, and (2) we keep track of the burstiness of each flow, to get the current worst-case delay after each reservation. These concepts do not require changes to the protocol. In addition, we propose to (3) distribute the information about the idleSlopes and priority-to-queue mappings of neighboring devices to each node. This allows us to derive tighter bounds, as will be shown in Eq.~\eqref{equ:improved-queue-arrival}.

\textbf{Setup.}
Before reserving flows in the network, each \ac{CBS} output queue is pre-configured with an idleSlope $idSl^{u,v,p}$ and a maximum per-hop delay budget $\overline{D}^{u,v,p}$.

\textbf{Flow Reservation.} During the stream advertisement, each bridge updates the maximum accumulated latency field of the reservation protocol $\overline{\mathcal{D}}_f^{u,v,p}$ using $\overline{D}^{u,v,p}$. Respectively, for the minimum latencies. When received by the listener, the maximum accumulated latency is checked against the flow's end-to-end deadline requirement.

During the listener subscription, each hop verifies that the (currently computed) worst-case delay $D^{u,v,p}$ after the reservation of the new flow does not violate the (pre-configured) guaranteed delay $\overline{D}^{u,v,p}$. 
If the guaranteed delay is not surpassed, the reservation of the new flow is accepted.

\subsection{Worst-Case Delay Computation $\mathbf{D^{u,v,p}}$}
To derive the worst-case delay for a flow $f$, we apply the mathematical framework \ac{NC} to determine guaranteed upper bounds on transmission and queuing delays in networks. Let us denote the cumulative input and output functions for each queue in the network as $R(t)$ and $R^*(t)$ respectively. We can define upper bounds on the input and output at each queue using arrival, shaping, and service curves.
\begin{definition}\label{def:arrival}
  (Arrival Curve~\cite{boudec_network_2012}): Let $\alpha(t)$ denote the upper bound on the arriving data over any period of length~$t$:
  \begin{equation}
    \forall s \leq t: R(t) - R(s) \leq \alpha(t-s)
\end{equation}
\end{definition}
\begin{definition}
  (Service Curve~\cite{boudec_network_2012}): Let $\beta(t)$ be the minimum service of a queue during any period of length~$t$:
  \begin{equation}
    R^*(t) \ge \inf_{s\le t}\big\{R(s) + \beta(t-s) \big\}
\end{equation}
\end{definition}
\begin{definition}
\label{def:shaper}
  (Shaping Curve~\cite{bouillard_deterministic_2018}) A queue offers a shaping curve \(\sigma\) if its output $R^*$ has \(\sigma\) as an arrival curve: 
\begin{equation}
    \forall s \leq t: R^*(t) - R^*(s) \leq \sigma(t-s)
\end{equation}
\end{definition}
As a result, the shaping curve of a queue also represents an arrival curve for subsequent queues.
\begin{definition}
  (Per-Hop Delay Bound~\cite{boudec_network_2012}) Let $h(f(t),g(t))$ be the maximum horizontal distance between $f(t)$ and $g(t)$. The worst-case delay at queue $(u,v,p)$ is:
  \begin{equation}
\label{equ:per-hop-delay}
   D^{u,v,p} \le h\Big(\alpha^{u,v,p}(t), \beta^{u,v,p}(t)\Big)
\end{equation}
\end{definition}
Intuitively, this is the maximum time between the arrival and departure of traffic in queue $(u,v,p)$. We use this bound after a new reservation to ensure that $D^{u,v,p}\le \overline{D}^{u,v,p}$. The arrival at each queue is the sum of arrivals of all flows, from each input link, which can be defined by
\begin{equation}\label{equ:arrival_link_shaped}
   \alpha^{u,v,p}(t) = \sum_{l \in \mathcal{L}^-_{u,v,p}} \Big(\alpha^{u,v,p}_l(t) \wedge \sigma_l(t) \Big)
\end{equation}
with $x \wedge y = \min\{x, y\}$ and
\begin{equation}
\label{equ:queue-arrival}
   \alpha^{u,v,p}_l(t) = \sum_{f\in F^{u,v,p}_l} \alpha_f(t+\Delta d^{u,v,p}_f)
\end{equation}
\begin{equation}
\label{equ:max-min-delay}
    \Delta d^{u,v,p}_f =  \overline{\mathcal{D}}_f^{u,v,p} - \underline{\mathcal{D}}_f^{u,v,p}.
\end{equation}
As a result, $\Delta d^{u,v,p}_f$ is the maximum accumulated delay difference of a flow $f$ before queue $(u,v,p)$, defined by:
\begin{equation}
    \overline{\mathcal{D}}_f^{u,v,p} = \sum_{q \in \Phi_f^{u,v,p}} \overline{D}^q
,\quad \text{and}\quad
    \underline{\mathcal{D}}_f^{u,v,p} = |\Phi_f^{u,v,p}| \cdot \frac{l_{f}}{C}
\end{equation}
If we additionally provide each bridge with the information about the \ac{CBS} shaping behavior of their neighboring nodes, we can improve Eq.~\eqref{equ:queue-arrival}. Instead of only shaping all flows from the same preceding link, we can additionally shape all flows from the same preceding queue. Then, Eq.~\eqref{equ:queue-arrival} improves to:
\begin{equation}
\label{equ:improved-queue-arrival}
   \alpha^{u,v,p}_l(t) = \sum_{q \in Q^-_l} \Big( \big( \sum_{f\in F^{u,v,p}_{l,q}} \alpha_f(t+\Delta d_f^{u,v,p}) \big) \wedge \sigma^q_{CBS}(t) \Big)
\end{equation}
If the number of queues is the same for each bridge, the preceding queue and priority are the same. However, since \ac{CBS} allows shared queues for multiple priorities, we have to shape flows with regard to their priority-to-queue mapping. 
We will discuss how this improvement can be introduced to the standard in Section~\ref{sec:protocol-adaption}. Fig.~\ref{fig:example-approach} illustrates the intuitive approach behind the above calculations. The definition for $\alpha_f$, $\beta^{u,v,p}$, $\sigma_l$, and $\sigma^q_{CBS}$ will be introduced in the next section, where we also refer to their original proofs.

\begin{figure}[t]
	\centering
	\includegraphics[width=0.9\linewidth]{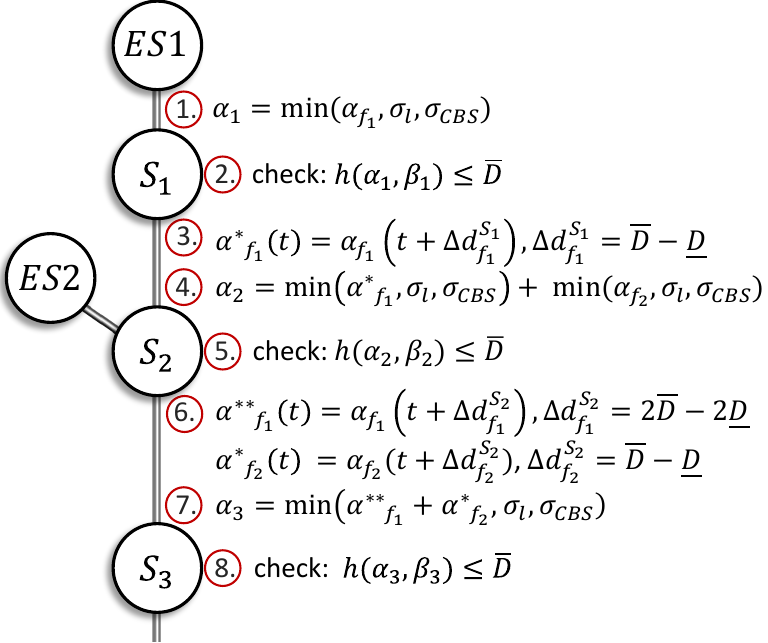}
	\caption{Iterative reservation procedure. We assume flow $f_1$ from $ES1$ and flow $f_2$ from $ES2$ to traverse both $S_3$. Each bridge is configured with the same maximum delay bound.}
	\label{fig:example-approach}
\end{figure}

\subsection{Proof}

In \ac{CBS} networks, the upper bound on the transmission of flows at their source has been derived and proven for both, periodic and aperiodic flows in~\cite{framework}. For simplicity, we only provide the results for periodic flows below.
\begin{theorem}[Flow Arrival {\cite[Theorem~1]{framework}}] 
\label{theo:arrival}
Each flow~$f$ with \ac{TSpec} of \ac{CMI}, \ac{MIF}, and \ac{MFS}, admits as arrival curve:
\begin{equation}
\label{equ:per-flow-arrival}
    \alpha_f(t) = m \cdot \Big\lceil \frac{t}{\mathit{CMI}} \Big\rceil 
\end{equation}
where $m=\mathit{MFS}\cdot \mathit{MIF}\cdot 8$. 
\end{theorem} 
This results in a staircase function, which omits the necessity of the ceiling function, as has been required in Eq.~\eqref{equ:proof_b_f}. 
\begin{theorem}[Flow Aggregate {\cite[p.~116]{cruz_calculus_1991}}]
\label{theo:flow_aggregate}
Assume flow $f_1$ and $f_2$ with $\alpha_1$ and $\alpha_2$ as their arrival curve, respectively. When scheduled by the same queue, we say that they arrive as aggregate. Then, the aggregate has $\alpha_1 + \alpha_2$ as arrival curve. 
\end{theorem}
\begin{theorem}[Output from Delay]
\label{theo:output_delay}
The output arrival of a flow traversing a queue with a maximum delay of $\overline{D}$ and a minimum delay $\underline{D}$ can be obtained by
\begin{equation}
    \alpha^*(t) = \alpha(t+\overline{D}-\underline{D}).
\end{equation}
\end{theorem}

\begin{proof}
Using the definition of arrival curves, $\underline{D}$, and $\overline{D}$, we can define:
\begin{align}
R^*(y)-R^*(x) & \leq R(y-\underline{D})-R(x-\overline{D})  \label{equ:proof1}\\
& \leq \alpha(y-\underline{D}-x+\overline{D}) \label{equ:proof2} \\
& =\alpha^*(y-x)\label{equ:proof3}
\end{align}
Setting $(y-x)=t$ concludes the proof.
\end{proof}
We illustrate the above proof using Fig.~\ref{fig:minmaxdelayproof}. Assume $\overline{D}=\underline{D}=0$, then $R^*(t)=R(t)$, 
meaning every bit which enters a queue directly leaves. Now assume $\overline{D}>0$, thus, all data which entered $\overline{D}$ time units before could leave as well:
\begin{equation}
    R^*(y)-R^*(x) \leq R(y)-R(x-\overline{D})
\end{equation}
This has already been proven in~\cite[Theorem~2.1]{cruz_calculus_1991}.
With $\underline{D}>0$, data that entered after $y-\underline{D}$ cannot leave by definition of  $\underline{D}$. As a result, it is subtracted from our observation interval, which leads to~\eqref{equ:proof1}. Equation~\eqref{equ:proof2} and~\eqref{equ:proof3} follow from the definition of arrival curves.
\begin{figure}[t]
		\centering\includegraphics[width=0.75\linewidth]{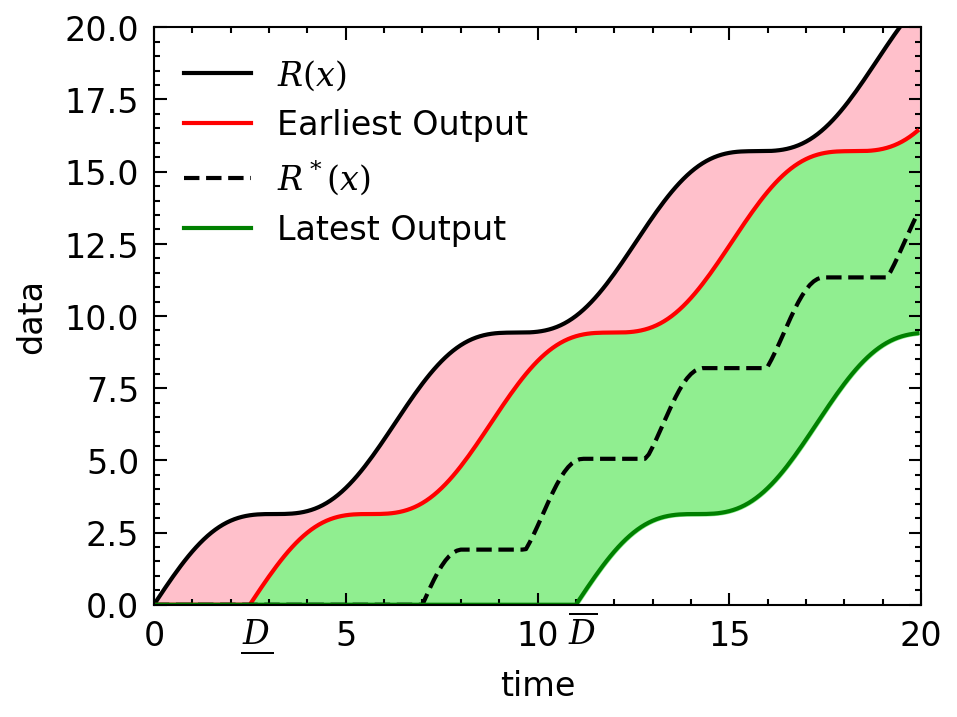}
		\caption{Output with known minimum and maximum delay. $R^*(x)$ is upper and lower bounded within the green area.}\label{fig:minmaxdelayproof}
\end{figure}
\begin{theorem}[Delay Bound {\cite[Theorem~1.4.2]{boudec_network_2012}}]
\label{theo:delay_bound}
The worst-case delay $D$ for arrival curve $\alpha$ in a queue with service curve $\beta$ is bounded by the maximum horizontal distance: $D \le h(\alpha, \beta)$
\end{theorem}
\begin{theorem}[Aggregate Delay Bound {\cite[Section~2.1.1]{boudec_network_2012}}]
\label{theo:aggregate_bound}
The worst-case delay bounds are the same for all aggregated flows and are computed with the aggregated arrival curve. 
\end{theorem}
\begin{theorem}[Service Curve{~\cite[p.~64]{queck}}]
\label{theo:service_curve}
The service curve for \ac{CBS} queues is a rate-latency curve with the following form:
\begin{equation}\label{equ:cbsservice}
	\beta_{R,T}(t) = R \cdot \big[t-T\big]^+  
\end{equation}
For queue $(u,v,p)$ the rate $R$ and latency $T$ are defined as:
\begin{equation}
	\begin{split}
		R^{u,v,p}  = idSl^{u,v,p},\qquad
		T^{u,v,p}  = \frac{c_{max}^{u,v,p}}{idSl^{u,v,p}}
	\end{split}
\end{equation}
\end{theorem}
\begin{theorem}[Credit Bounds{~\cite[Theorem~2~and~3]{zhao_latency_2020}}]
\label{theo:credit_bounds}
The maximum and minimum \ac{CBS} credits are defined as:
\begin{equation}\label{equ:credit}
	\begin{split}
		& c_{max}^{u,v,p} = idSl^{u,v,p} \cdot \frac{\sum_{i=p+1}^{Q}c_{min}^{u,v,i} - L^{u,v,<p}}{\sum_{i=p+1}^{Q}idSl^{u,v,i}-C} \\
		~ \text{and } ~ & c_{min}^{u,v,p} = sdSl^{u,v,p} \cdot \frac{L^{u,v,p}}{C}
	\end{split}
\end{equation}
\end{theorem}
\begin{theorem}[Link Shaping Curve {\cite[Section~7.1]{kellerer_network_2016}}] 
\label{theo:link_shaper_curve}
 The link shaping curve in store-and-forward networks is:
 \begin{equation}
	\sigma^{u,v,p}_l(t) = L^{u,v,p} + Ct
\end{equation}
\end{theorem}

\begin{theorem}[CBS Shaping Curve {\cite[Theorem~5~and~9]{framework}}] 
\label{theo:cbs_shaper_curve}
As $idSl^{u,v,p}$ also denotes the maximum sustained output rate, each \ac{CBS} queue $(u,v,p)$ offers a shaping curve of:
\begin {equation}\label{equ:shaperCBS}
\sigma^{u,v,p}_{CBS}(t) = idSl^{u,v,p} \cdot \bigg(t+\frac{c_{max}^{u,v,p}-c_{min}^{u,v,p}}{idSl^{u,v,p}}\bigg) + L^{u,v,p}
\end {equation} 
\end{theorem}

For Theorem~\ref{theo:link_shaper_curve} and~\ref{theo:cbs_shaper_curve}, see also~\cite[Eq.~(21)]{zhao_quantitative_2022}.
To sum up, we showed that we can use the local maximum and the minimum delay budget ($\overline{D}$ and $\underline{D}$) to compute the arrival curves at each hop. 
As the local delay budgets do not change for new reservations, the arrival curves for existing reservations at each node do not change as well; we only have to validate new reservations. In other words, nodes do not need to be informed about new reservations of flows in other parts of the network. The local delay budget is validated using the current worst-case delay, as presented in Eq.~\eqref{equ:per-hop-delay}.
The proof for the current worst-case delay can be derived as follows:
\begin{proof}[Proof of Eq.~\eqref{equ:per-hop-delay}]
Using Eq.~\eqref{equ:max-min-delay} and Theorem~\ref{theo:output_delay}, we can show that
\begin{equation}
    \alpha_f^{u,v,p}(t) = \alpha_f(t+\Delta d^{u,v,p}_f).
\end{equation}
With Theorem~\ref{theo:flow_aggregate} about the flow aggregate, this results in Eq.~\eqref{equ:queue-arrival}.
By Def.~\ref{def:shaper}, we know that outputs can be upper bounded with shaping curves, resulting in $\min(\alpha^*(t), \sigma(t))$. Thereby, the output of a CBS queue has Theorem~\ref{theo:cbs_shaper_curve} as shaping curve, while traversing a link adds Theorem~\ref{theo:link_shaper_curve} as shaping curve. With Definition~\ref{def:arrival}, we can prove Eq.~\eqref{equ:arrival_link_shaped} and, correspondingly, Eq.~\eqref{equ:improved-queue-arrival}. Finally, Theorem~\ref{theo:aggregate_bound} and~\ref{theo:delay_bound} result in the per-hop delay bound of Eq.~\eqref{equ:per-hop-delay}, which we can evaluate using Theorem~\ref{theo:service_curve}. 
\end{proof}

\subsection{Protocol Adaption}
\label{sec:protocol-adaption}
For the integration of \ac{SSRP} into the established layer~2
reservation, all variables that are required to calculate $D^{u,v,p}$ in Eq.\,\eqref{equ:per-hop-delay} - that are not locally
available - must be obtained through the reservation protocol.
The service curve~$\beta$ is only based on locally available information (from Eq.\,\ref{equ:cbsservice}), such as the configured idleSlopes for each priority.
For the arrival curve~$\alpha$, the individual arrival curves~$\alpha_f$ of each flow~$f$ are required (Eq.\,\ref{equ:arrival_link_shaped}).
These can be calculated based on the \acp{TSpec} from \ac{SRP} and \ac{RAP}, as shown in Eq.\,\eqref{equ:per-flow-arrival}.
In addition to that, the term $\Delta d_f$ can be obtained by using the accumulated latency fields in both \ac{SRP} and \ac{RAP}.
Note that SRP does not include an accumulated minimum latency, but 0 can always be used as a worst-case estimate.
In addition, SRP must be amended to validate that the pre-configured local guarantees are never exceeded.
Otherwise, the accumulated maximum latency is not safe to use.

When applying the improved formula in Eq.\,\eqref{equ:improved-queue-arrival},
the shaping curve~$\sigma_{CBS}^q$ of the previous queue~$q$ must be known.
As shown in Eq.\,\eqref{equ:credit} and~\eqref{equ:shaperCBS},
this curve relies on the idleSlope of the previous shaper and the maximum and minimum
frame sizes of that previous queue. Depending on the network, the applied priority-to-queue mapping might also be
required in order to compute~$\sigma_{CBS}^q$.
This type of information about other switches is not yet available in neither SRP nor RAP.
Currently, RAP contributors consider sharing shaper-specific information with neighbors a priori
in order to support heterogeneous network configurations.
Hence, this type of information could also be used to implement the improved latency formula.

\section{Evaluation}
\label{sec:evaluation}

In this section, we use simulation to show that the state-of-the-art approaches cannot provide an upper bound for all scenarios with a counterexample. We also demonstrate that our method benefits from additional information about the CBS shaping of neighboring nodes and how one could apply it to a real-world network.
To offer precise evaluation, we assumed a priori knowledge of the arriving flows, reflecting an offline configuration before the network setup. See \cite{grigorjew_ML_2020} for per-hop delay bounds in online configurations during the runtime of the network.

\subsection{Comparison of Existing Approaches}
\label{sec:eval_comparison}
First, we compare all delay bounds with a simulation. We implemented the network in Fig. \ref{fig:eval_comparison_network} in OMNeT++ with INET~4.4. All talkers send packets with a \ac{CMI} of \SI{125}{\micro\second} to the same listener. We varied the number of talkers and, consequently, the number of input links on the last switch. We adjust the packet size to keep the utilization of the link between the last switch and the listener at 75\% (= idleSlope) for traffic with priority~7. Additional sources inject best-effort packets into the outgoing queues just before a priority~7 packet arrives, like in Section~\ref{sec:infinite}. This happens at five stages on all input links, which leads to a burst of packets at the last switch.

Fig.~\ref{fig:eval_comparison_results} shows the resulting maximum queuing delay in the last switch. Since the measured maximum queuing delays of the simulation exceed the upper bounds of the standards at one or more points, we can show that they do not provide guaranteed delay bounds. On the other hand, \ac{SSRP} results in a boundary that is always higher than the measured delay. Note that the minimum Ethernet packet size limits the maximum number of input links to 13, thus, the bounds of \ac{SSRP} will not decrease further in Fig.~\ref{fig:eval_comparison_results}.

\begin{figure}[t]
		\centering\includegraphics[width=0.9\linewidth]{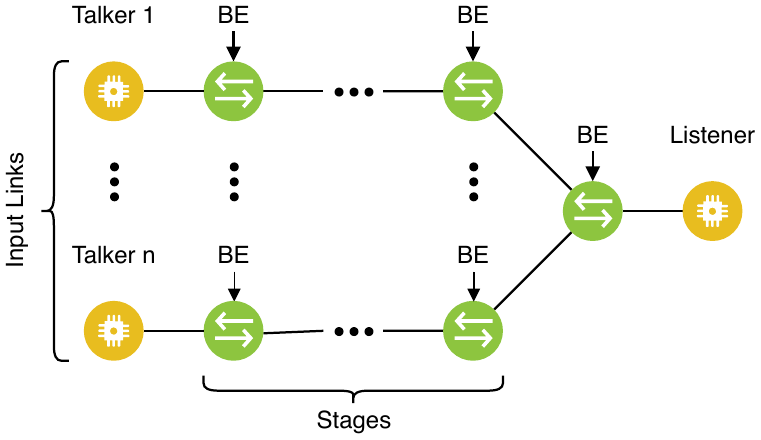}
		\caption{Network used to compare the different approaches.}\label{fig:eval_comparison_network}
\end{figure}

\subsection{Credit-Based Shaping}

As introduced in Eq.~\eqref{equ:improved-queue-arrival}, \ac{SSRP} can use information about the \ac{CBS} shaping of neighboring nodes to improve the latency bound. This evaluation demonstrates that this knowledge improves the number of successful reservations. We used a line topology with six switches, a talker connected to each, and a listener attached to the sixth switch. We then added flows from a randomly selected talker to the same listener. This test adds new flows until the sum of all maximum delays exceeds a maximum end-to-end delay. We used \SI[per-mode=symbol]{1}{\giga\bit\per\second} links with an idleSlope of 75\%, 128~Byte packets, and \ac{CMI} = \SI{125}{\micro\second}. Therefore, the link rate limits the maximum number of flows to 91. We repeated this test 1000 times. Fig.~\ref{fig:eval_synthetic} shows the mean number of flows with one standard deviation for a given end-to-end delay. With the knowledge about the shaping of neighboring nodes, up to 20\% more reservations are possible.

\begin{figure}[t]
		\centering\includegraphics[width=0.9\linewidth]{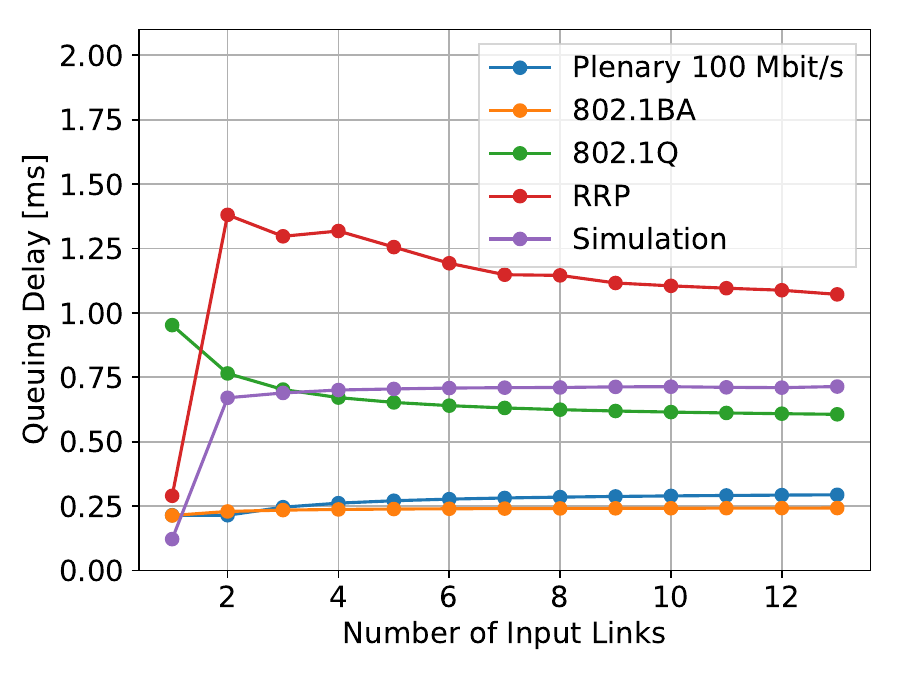}
		\caption{Maximum queuing delay at the last switch.}\label{fig:eval_comparison_results}
\end{figure}

\subsection{Industrial Use-Case}

To evaluate \ac{SSRP} in a practical use case, we simulate an industrial network based on PROFINET. Fig.~\ref{fig:eval_profinet} shows the network topology with \SI[per-mode=symbol]{100}{\mega\bit\per\second} links based on~\cite[Fig.~6-18]{profinet}. Each line consists of three I/O devices with integrated switches, typically representing a machine component. All I/O devices talk to a \ac{PLC} via a central switch that connects to all lines. We want to simulate this network for the maximum number of I/Os. We found this limit by increasing the number of lines until the sum of the maximum delays along the paths exceeded the cycle time of \SI{1}{\milli\second}. In this example, seven lines, i.e., 21 I/Os, each with a \ac{CMI} of \SI{125}{\micro\second} and \ac{MFS} of 110~Byte\footnote{The two additional bytes were necessary to avoid an issue with the OMNeT++ simulation.} is the maximum \cite{profinet}. The load on the link to the PLC is \SI[per-mode=symbol]{18.5}{\mega\bit\per\second} and within the acceptable range \cite{profinet}. We added a \ac{NRT} traffic source and sink, e.g., a surveillance camera and monitor, at the end of the first and last line respectively. As a result, \ac{NRT} packets delay PROFINET traffic, 
which is a good stress test for \ac{SSRP}. The \ac{NRT} source sends maximum length packets exponentially distributed with a mean of \SI{300}{\micro\second} resulting in 40\% additional load \cite{profinet}.

\begin{figure}[t]
		\centering\includegraphics[width=0.9\linewidth]{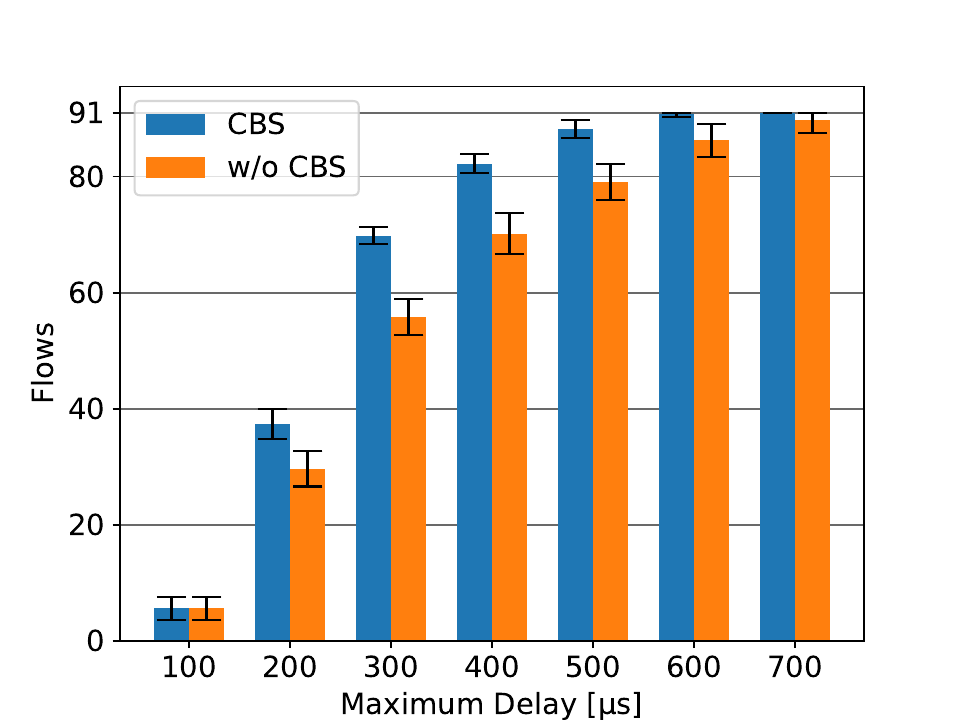}
		\caption{Number of flows for a maximum delay, with and without information about the neighbors' CBS shaping. 
  }\label{fig:eval_synthetic}
\end{figure}

Table \ref{tab:eval_profinet} shows the queuing delays on each hop on the first line and the central switch. For the first switches, the delays are very close, since it is likely to encounter a situation where both a PROFINET and \ac{NRT} packet arrive at an unfavorable timing. The considerable difference in the central switch is due to the absence of \ac{NRT} traffic. Redirecting the \ac{NRT} traffic to the PLC increases the delay to \SI{344}{\micro\second}. As we can see, the measured delays in this practical task can be close to the limits computed with \ac{SSRP} but never exceed them.

\textbf{Discussion:} The differences between delay bounds and the measured delay in Table~\ref{tab:eval_profinet} are caused by the necessary assumption in \ac{SSRP} that every switch can possibly use its full delay guarantee. A gap leaves the opportunity for devices to accept additional flow reservations without validating existing reservations. Closing this gap would require significantly more signaling in decentralized systems, as an increase of queuing delay in one device results in cascading latencies over the network, possibly leading to a rejection of already reserved flows. 

\section{Conclusion}
\label{sec:conclusion}
In this paper, we examine the reliability of latency bounds in Time-Sensitive Networking standards for Credit-Based Shaper networks, and showed that the latency bounds of the standards are not reliable due to the potential for infinite burstiness of interfering traffic. Our simulations could also show that the latency bounds are missed in real networks.

To address this issue, we propose a resource reservation process that includes a per-hop latency budget for Credit-Based Shaper networks. Our approach is designed to provide deterministic latency guarantees, using only bridge-local information. 
The traffic specification requirements for resource reservation, including burst, rate, and accumulated minimum and maximum latency, remain consistent with existing protocols. However, we propose a validation of the actual worst-case delay at each hop to ensure that it does not exceed the maximum delay budget. Additionally, we can consider the shaping of previous queues in the reservation. This shaping can improve the delay calculation, resulting in an increase of up to 20\% in the number of flows in our simulations. In a formal proof, we demonstrate that our approach eliminates the need for re-configuration of the network after new reservations and that packets do not exceed their bounds. Our simulations demonstrate the practicality of our approach for industrial networks, as it is characterized by low pessimism and provides reliable guarantees.

\begin{figure}[t]
		\centering\includegraphics[width=0.9\linewidth]{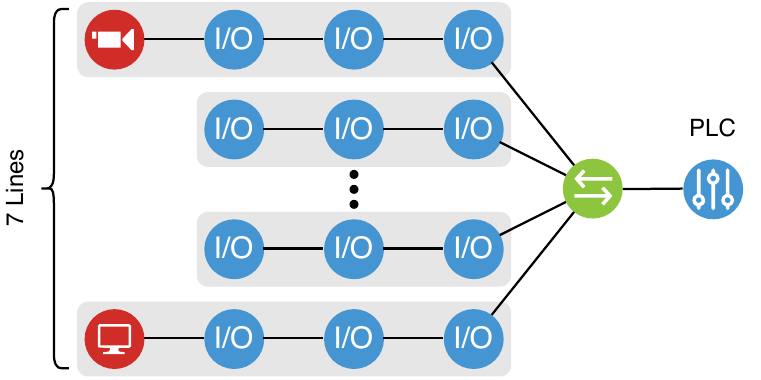}
		\caption{Network topology based on PROFINET \cite{profinet}.}
  \label{fig:eval_profinet}
\end{figure}

\begin{table}[t]
\centering
\caption{PROFINET results}\vspace{-2mm}
\begin{tabular}{|l|l|l|}
\hline
Switch & RRP [\SI{}{\micro\second}] & Simulation [\SI{}{\micro\second}] \\ \hline\hline
1 & 141 & 123 \\
2 & 159 & 123 \\
3 & 177 & 124 \\
Central & 494 & 234 \\
  \hline
\end{tabular}

\label{tab:eval_profinet}
\end{table}

In future research, we aim to investigate centralized and decentralized configurations within the same network, by the incorporation of the central approach presented in~\cite{maile_journal_2022}. Furthermore, we intend to optimize the per-hop delay bounds to maximize the number of successful reservations through the use of heuristics and machine learning techniques. Additionally, we want to investigate the effect of bottlenecks in the network and how they can be considered during reservation.

\begin{acks}
This research was partly funded by the Bavarian Ministry
of Economic Affairs, Regional Development and Energy under
grant number DIK0250/02, project KOSINU5.
\end{acks}

\bibliographystyle{ACM-Reference-Format}
\bibliography{sample-sigplan}

\appendix

\begin{acronym}
\acro{WCRT}[WCRT]{worst-case response time}
\acro{SFD}[SFD]{start frame delimiter}
\acro{IFG}[IFG]{inter-frame gap}
\acro{IPG}[IPG]{inter-packet gap}
\acro{SP}[SP]{Strict Priority}
\acro{QoS}[QoS]{Quality of Service}
\acro{TSN}[TSN]{Time-Sensitive Networking}
\acro{NC}[NC]{Network Calculus}
\acro{CMI}[CMI]{Class Measurement Interval}
\acro{MFS}[MFS]{Maximum Frame Size}
\acro{MIF}[MIF]{Maximum Interval Frame}
\acro{AV}[AV]{Audio and Video}
\acro{AVB}[AVB]{Audio and Video Briding}
\acro{SRP}[SRP]{Stream Reservation Protocol}
\acro{SR}[SR]{Stream Reservation}
\acro{CBS}[CBS]{Credit-Based Shaper}
\acro{BE}[BE]{Best-Effort}
\acro{ES}[ES]{End-Station}
\acro{RAP}[RAP]{Resource Allocation Protocol}
\acro{FIFO}[FIFO]{First-In-First-Out}
\acro{FoI}[FoI]{Flow of Interest}
\acro{NRT}[NRT]{Non-Real-Time}
\acro{PLC}[PLC]{Programmable Logic Controller}
\acro{SSRP}[RRP]{Reliable Reservation Protocol}
\acro{RAP}[RAP]{Resource Allocation Protocol}
\acro{NMS}[NMS]{Network Management System}
\acro{TSpec}[TSpec]{Traffic Specification}
\end{acronym}

\end{document}